\documentclass[12pt, a4paper, reqno]{amsart}

\usepackage{amssymb, amsrefs}
\usepackage{hyperref}
\usepackage{setspace}

\theoremstyle{plain}
\newtheorem{thm}{Theorem}[section]

\newtheorem{cor}[thm]{Corollary}
\newtheorem{lem}[thm]{Lemma}

\newcommand{\R}{\mathbb{R}}
\newcommand{\N}{\mathbb{N}}
\newcommand{\Sch}{\mathbb{S}}
\newcommand{\h}{\hbar}
\newcommand{\cU}{\mathcal{U}}
\newcommand{\cZ}{\mathcal{Z}}
\newcommand{\cL}{\mathcal{L}}
\newcommand{\e}{\epsilon}
\newcommand{\Op}{\mathrm{Op}_{\h}^w}
\newcommand{\supp}{\mathrm{supp \,}}
\newcommand{\vol}{\operatorname{Vol}}
\newcommand{\bra}{\left\langle}
\newcommand{\ket}{\right\rangle}
\newcommand{\diam}{\operatorname{diam}}
\newcommand{\Id}{\operatorname{Id}}

\title[Scarred quasimodes on translation surfaces]{Scarred quasimodes on \\ translation surfaces}

\date\today

\begin{document}

\begin{abstract}
Rational polygonal billiards are one of the key models among the larger class of pseudo-integrable billiards. Their billiard flow may be lifted to the geodesic flow on a translation surface. Whereas such classical billiards have been much studied in the literature, the analogous quantum billiards have received much less attention. 

This paper is concerned with a conjecture of Bogomolny and Schmit who proposed in 2004 that the eigenfunctions of the Laplacian on rational polygonal billiards ought to become localized along a finite number of vectors in momentum space, as the eigenvalue tends to infinity. 

For any given momentum vector $\xi_0\in\Sch^1$ we construct a continuous family of quasimodes which gives rise to a semi-classical measure whose projection on momentum space is supported on the orbit $D\xi_0$, where $D$ denotes the dihedral group associated with the rational polygon.
\end{abstract}

\author{Omer Friedland}
\address{Institut de Math\'ematiques de Jussieu - Paris Rive Gauche, Sorbonne Universit\'e - Campus Pierre et Marie Curie, 4 place Jussieu, 75252 Paris, France.}
\email{omer.friedland@imj-prg.fr}

\author{Henrik Uebersch\"ar}
\address{Institut de Math\'ematiques de Jussieu - Paris Rive Gauche, Sorbonne Universit\'e - Campus Pierre et Marie Curie, 4 place Jussieu, 75252 Paris, France.}
\email{henrik.ueberschar@imj-prg.fr}

\thanks{H. U. was supported by the grant ANR-17-CE40-0011-01 of the French National Research Agency ANR (project SpInQS)}

\maketitle

\section{Introduction}

Let $P\subset \R^2$ be a rational polygon, i.e. all its angles are rational multiples of $\pi$. Let us consider the motion of a point mass in the associated classical polygonal billiard. Due to the rationality of the angles, it is a simple fact that in momentum space the point mass can only assume a finite number of different directions (see \cite[Ch. 7]{T05}).

Therefore, from the point of view of semiclassical analysis, it is a natural question to ask, how the eigenfunctions of an associated quantum billiard are distributed in momentum space in the semiclassical limit, i.e. as the eigenvalue tends to infinity. Indeed, Bogomolny and Schmit \cite{BS04} conjectured that, in the semiclassical limit, the Laplace eigenfunctions on a rational polygon should become localized in momentum space along a finite number of momentum vectors -- a phenomenon that has been dubbed ``superscars''. 
A complementary question about the distribution of the eigenfunctions in the semiclassical limit has been answered by Marklof and Rudnick \cite{MR12} who proved that a full density subsequence of eigenfunctions equidistributes in position space.

It is well-known that the billiard flow on a rational polygon $P$ may be lifted to the geodesic flow on an associated translation surface $Q$, which is a flat surface with conical singularities (this is a classical construction, see e.g. \cite{MT02, T05, ZK75}).

Let us consider the Laplace-Beltrami operator, $\Delta=\partial_x^2+\partial_y^2$ in local Euclidean coordinates, on the flat surface $Q$. We are interested in Laplace eigenfunctions $\psi_\lambda$, with eigenvalue $\lambda>0$
$$
(\Delta+\lambda)\psi_\lambda=0. 
$$

Denote the set of conical singularities by $C$ and let $Q_0:=Q\setminus C$. The pseudo-differential calculus permits to associate with each $L^2$-normalized function $\psi$, a Wigner distribution $d\mu_\psi(x,\xi)$ on the unit co-tangent bundle $S^*Q_0$ which, for any classical observable $a\in C^\infty_c(S^*Q_0)$ is defined by the duality (see Section \ref{secL2} for a detailed explanation)
$$
\bra \Op(a)\psi,\psi\ket_{L^2(Q)} = \int_{S^*Q_0}a(x,\xi)d\mu_\psi(x,\xi)~, \quad \text{where } \|\psi\|_2=1,
$$
where $\Op(a)$ denotes the standard Weyl quantization.

We denote the restriction of the Wigner distribution to momentum space by 
$$
d\mu_\psi(\xi)=\int_Q d\mu_\psi(x,\xi). 
$$


Our goal is to construct approximate eigenfunctions for the Laplacian, known as quasimodes, and prove that they satisfy the above conjecture. 
%
Let $\Psi_\lambda\in C^2(Q)$ so that 
$$
\frac{\|(\Delta+\lambda)\Psi_\lambda\|_{L^2(Q)}}{\|\Psi_\lambda\|_{L^2(Q)}}=O(\lambda^\delta) .
$$
We say that $\Psi_\lambda$ is a quasimode with quasienergy $\lambda$ and spectral width $\lambda^\delta$ of the Laplacian on $Q$. 
%
We have the following theorem.

\begin{thm} \label{thm-Q}
Let $\xi_0\in\Sch$. Then for small enough $\e>0$ there exists a continuous family of quasimodes $\{\Psi_{\lambda}\}_{\lambda>0}$ for the Laplacian on $Q$ of spectral width $O(\lambda^{3/8+\e})$ so that 
$$
d\mu_{\Psi_{\lambda}}(\xi) \xrightarrow{w*} \delta(\xi-\xi_0)~,\quad \text{as } \lambda\to\infty.
$$
\end{thm}

As a corollary of the theorem we may construct a sequence of quasimodes for the Neumann Laplacian on a rational polygon $P$. Indeed, any rational polygon $P$ may be unfolded to a translation surface $Q$ under the action of the dihedral group $D$ of $P$ (see e.g. \cite[Section 1.5]{MT02}). Given a quasimode $\Psi_\lambda$ on $Q$, we may construct a quasimode on $P$ by the method of images, 
$$
\Psi_\lambda^P(x)=\sum_{g\in D}\Psi_\lambda(gx).
$$

\begin{cor} \label{cor-P}
Let $\xi_0\in\Sch$. Then for small enough $\e>0$ there exists a continuous family of quasimodes $\{\Psi_{\lambda}^P\}_{\lambda>0}$ for the Neumann Laplacian on $P$ of spectral width $O(\lambda^{3/8+\e})$ so that 
$$
d\mu_{\Psi_{\lambda}^P}(\xi) \xrightarrow{w*} \frac{1}{|D|}\sum_{g\in D}\delta(\xi-g\xi_0)~, \quad \text{as } \lambda\to\infty, 
$$
where $D$ is the dihedral group of $P$.
\end{cor}

\section{Construction of the quasimode}

The main idea of the construction of our quasimode is to take a coherent initial state which is localized in position and in momentum and average the evolved state over time. This construction is analogous, for example, to the one employed by Eswarathasan and Nonnenmacher \cite{EN} in the case of compact manifolds with a hyperbolic periodic orbit. In our case, since we are dealing with flat surfaces with conical singularities, we may average over much larger time scales. The length of the time interval is constrained by the condition that the evolved wave packet must remain inside embedded metric cylinders which correspond to parallel families of periodic orbits, and thus avoid the conical singularities.

Let us, first, introduce a coherent state on the Euclidean plane in order to motivate our construction. Let $(x_0,\xi_0)\in\R^2\times\Sch^1$, $\h=\lambda^{-1/2}$ and denote
$$
\varphi_0(x) = \sqrt{\frac{\pi}{\h}}\gamma(\frac{x-x_0}{\h^{1/2}})e^{i \frac{\xi_0\cdot x}{\h}}, 
$$
where $\gamma(x)=\frac{1}{2\pi}e^{-|x|^2/2}$. A simple calculation confirms that the state $\varphi_0$ is localized in position near $x_0$ on a scale $\h^{1/2}$ and in momentum near $\xi_0/\h$ on a scale $\h^{-1/2}$.

Let us now consider a translation surface $Q$ and take $x_0\in Q$. We choose a cutoff function $\chi\in C^\infty_c(\R_+)$ so that $\chi(x)=1$ for $x\le \tfrac{1}{2}$ and $\chi(x)=0$ for $x\ge 1$. For $\h$ small enough and $\e>0$, we may construct $\psi_0\in C^\infty(Q)$ so that
$$
\psi_0(x)=\chi(\frac{|x-x_0|}{\h^{1/2-\e}})\varphi_0(x)
$$
in local Euclidean coordinates.

In order to obtain a quasimode for the Laplacian we now average over time the evolved state $\cU_t\psi_0$, where $\cU_t=e^{it\Delta}$. Let $H\in C^\infty_c(\R)$ with $\supp H=[-T,T]$ and $T>0$ is a time-scale which depends on $\h$ and set
$$
\Psi_\lambda = \int_\R H(t)e^{i\lambda t}\cU_t\psi_0 dt = \int_\R H(t)e^{it(\Delta+\lambda)}\psi_0 dt.
$$

For simplicity we take $\widetilde{H}\in C^\infty_c(\R)$, with $\supp(\widetilde{H})\subset[-1,1]$, and set $H(\cdot)=\widetilde{H}(\cdot/T)$. Now observe that
$$
(\Delta+\lambda)\Psi_\lambda =\int_\R H(t)(\Delta+\lambda)e^{it(\Delta+\lambda)}\psi_0 dt = i \int_\R H'(t)e^{it(\Delta+\lambda)}\psi_0 dt.
$$

So, in order to prove that $\Psi_\lambda$ is a quasimode on $Q$ and to obtain a bound on its spectral width, we must obtain the asymptotics of the $L^2$-norm of 
$$
\Psi_\lambda(G)=\int_\R G(\frac{t}{T})e^{it(\Delta+\lambda)}\psi_0 dt, 
$$
where $G=\widetilde{H}$, $T^{-1}\widetilde{H}'$.


\section{Asymptotics of $L^2$-norms}\label{secL2}

Fix $G\in C^\infty_c(\R)$ with $\supp(G)\subset [-1,1]$. A change of variables yields
\begin{align}\label{L2id}
\|\Psi_\lambda(G)\|_{L^2(Q)}^2=\frac{1}{2}\int_\R g(v)e^{-iv\lambda} \bra \psi_0, \cU_v \psi_0 \ket dv
\end{align}
where 
$$
g(v)=\int_{\R}G(\frac{u+v}{2T})G(\frac{u-v}{2T})du.
$$

Note that, by construction, almost all of the mass (except for a proportion of order $\h^\infty$) of the Wigner distribution associated with the state $\psi_0$ lies inside the set 
$$
\Omega_0=B(x_0,\h^{1/2-\e}) \times B(\xi_0,\h^{1/2-\e}) \subset T^*Q.
$$

To see this, we recall the definition of the Weyl quantization on $\R^2$. Let $\psi\in\Sch(\R^2)$ and $a=a(x,\xi)\in C^\infty_c(\R^2\times\R^2)$. We define (see for instance \cite{Zw12})
$$
[\Op(a)\psi](x)=\frac{1}{(2\pi \h)^2}\int_{\R^2\times\R^2}e^{\frac{i}{\h}(x-y)\cdot\xi}a(\frac{x+y}{2},\xi)\psi(y)dy d\xi.
$$

If we choose an observable, which only depends on the position variable $a=a(x)$ we obtain the restriction of the Wigner distribution to position space $d\mu_{\psi}(x)=\|\psi\|_2^{-2}|\psi(x)|^2dx$. On the other hand, if we take $a=a(\xi)$ we find
$$
\frac{\bra\Op(a)\psi,\psi\ket}{\|\psi\|_2^2}=\frac{1}{\h^2\|\psi\|_2^2}\int_{\R^2}a(x)|\widehat{\psi}(\frac{\xi}{\h})|^2 d\xi, 
$$
where we denote $\widehat{\psi}(k)=\frac{1}{2\pi}\int_{\R^2}\psi(x)e^{-ik\cdot x}dx$.

It follows that the restriction of the Wigner distribution to momentum space is given by $d\mu_{\psi}(\xi)=\h^{-2}\|\psi\|_2^{-2}|\widehat{\psi}(\frac{\xi}{\h})|^2$. By construction, the state $\psi_0$ is localized near $\xi_0/\h$ on a scale of size $\h^{-1/2}$. This implies that $d\mu_{\psi_0}(\xi)$ must be localized near $\xi_0$ on a scale of size $\h^{1/2}$. Similarly, the restriction of the Wigner distribution to position space is localized near $x_0$ on a scale $\h^{1/2}$, which implies that almost all of its mass is concentrated inside the set $\Omega_0$.

We may choose a classical observable $a_0\in C^\infty_c (S^*Q_0)$ with $\supp a_0\subset \Omega_0$ so that 
$$
\int_{S^*Q_0} (1-a_0)^2 d\mu_{\psi_0}=O(\h^\infty)
$$
which implies 
\begin{align*}
& |\bra \psi_0, \cU_v (\Id-\Op(a_0))\psi_0 \ket_{L^2(Q)}| \leq \|(\Id-\Op(a_0))\psi_0\|_{L^2(Q)} \\
&=\bra\psi_0,(\Id-\Op(a_0))^2\psi_0\ket^{1/2} =(\int_{S^*Q_0}(1-a_0)^2d\mu_{\psi_0})^{1/2}\ = O(\h^\infty) .
\end{align*}

In view of the decomposition 
$$\psi_0=\Op(a_0)\psi_0+(\Id-\Op(a_0))\psi_0$$
we obtain, from a substitution in \eqref{L2id},
\begin{align} \label{integral-overlap}
\nonumber & \| \Psi_\lambda(G)\|_{L^2(Q)}^2 \\
& =\frac{1}{2}\int_\R g(v)e^{-iv\lambda} \bra \psi_0, \cU_v \Op(a_0)\psi_0 \ket_{L^2(Q)} dv+O(\h^\infty).
\end{align}

If we now choose an initial direction $\xi_0$ corresponding to an embedded metric cylinder $\cZ_{\xi_0}$ (a family of parallel periodic orbits) of length $L$ and width $\asymp 1/L$, 
then for 
$$
|v|\le T=O(\h^{3/4+\e/2})
$$
we have that $\phi_{v/\h}\Omega_0 \subset \cZ_{\xi_0}$, where $\phi_t$ denotes the geodesic flow on $Q$. 

Indeed, this follows from the observation that 
$$
\diam(\pi_Q(\phi_{v/\h}\Omega_0))\asymp \frac{|v|}{\h}\diam(\pi_{\Sch^1}\Omega_0)=|v|\h^{-1/2-\epsilon}
$$
where $\pi_X$ denotes the canonical projection on $X$. In order to remain inside the cylinder we must have $\diam(\pi_Q(\phi_{v/\h}\Omega_0))$ is smaller than the width, that is, $|v|\h^{-1/2-\epsilon}=O(\frac{1}{L})$. Moreover, we impose the condition $\frac{|v|}{\h}\leq L$ which ensures that our wave packet may not travel further than the length of the periodic cylinder. This then implies $|v|\h^{-1/2-\epsilon}=O(\frac{\h}{|v|})$ and thus $|v|=O(\h^{3/4+\e/2})$. In particular, the set $\bigcup_{|v|\le T} \phi_{v/\h}\Omega_0$ does not self-intersect on $Q$.

We may lift the metric cylinder $\cZ_{\xi_0}$ to the Euclidean plane by embedding it in a Euclidean cylinder $\widetilde{\cZ}_{\xi_0}\subset \bigsqcup_i Q_i\subset \R^2$ which lies inside a union of disjoint translates $Q_i=\tau_i Q$ where $\tau_i$ is the corresponding translation vector (see \cite{MT02}).
 
We may, therefore, apply the exact version of Egorov's theorem for the Weyl quantization on the Euclidean plane (cf. \cite[Ch. 4]{Ma02}), since $\supp (a_0 \circ \phi_{-v/\h})\subset \phi_{v/\h}\Omega_0$ 
$$
\cU_v\Op(a_0)\psi_0=\Op(a_0\circ \phi_{-v/\h})\cU_v \psi_0.
$$

This means that the propagation of the cutoff state $\Op(a_0)\psi_0$ may simply be lifted to the cover $\bigsqcup_i Q_i$ and then projected back to $Q$. Since the set $\{\phi_{v/\h}\Omega_0\}_{|v|\le T}$ has no self-intersections we find that the overlap $\bra \psi_0, \cU_v \Op(a_0)\psi_0 \ket$, i.e. the integrand in \eqref{integral-overlap}, vanishes outside an interval $|v|\leq\h$. In other words, $\supp (\psi_0)$ and $\supp(a_0\circ \phi_{-v/\h})$ may trivially overlap, by construction, only around $v=0$. So the integral in \eqref{integral-overlap} collapses to an integral over an interval of length $\leq \h$. Up to an error of size $O(\h^\infty)$ this integral may be replaced with the corresponding Euclidean integral with the Gaussian initial state $\varphi_0$
\begin{align} \label{Eucl_L2}
\nonumber \|\Psi_\lambda(G)\|_{L^2(Q)}^2=&\frac{1}{2}\int_{|v|\le Ch} g(v)e^{-iv\lambda} \bra \psi_0, \cU_v \Op(a_0)\psi_0 \ket_{L^2(Q)} dv + O(\h^\infty)\\
=&\frac{1}{2}\int_{\R} g(v)e^{-iv\lambda} \bra \varphi_0, \cU_v \varphi_0 \ket_{L^2(\R^2)} dv + O(\h^\infty). 
\end{align}

The asymptotic of the Gaussian integral may now be computed explicitly in terms of special functions. This yields the following lemma.

\begin{lem}
We have $\|\Psi_\lambda(G)\|_{L^2(Q)}^2 \sim CT\h^{3/2}$, as $\h\to 0$.
\end{lem}

\begin{proof}
Let us introduce the Gaussian quasimode
$$
\Phi_\lambda(G)=\int_\R G(\frac{t}{T})e^{it(\Delta+\lambda)}\varphi_0 dt.
$$

In view of \eqref{Eucl_L2} we have $\|\Psi_\lambda(G)\|_{L^2(Q)}=\|\Phi_\lambda(G)\|_{L^2(\R^2)}+O(\h^\infty)$, and note that a change of variable yields $g(v)=T\widetilde{g}(v/T)$ where 
$$
\widetilde{g}(w)=\int_\R G(\frac{u+w}{w})G(\frac{u-w}{2})du
$$
and $\widetilde{g}\in C^\infty_c(\R)$ with $\supp\widetilde{g}\subset[-2,2]$.

If we now develop $\widetilde{g}$ into a Taylor series at $w=0$ we find, for some $c\in(-2,2)$,
\begin{align*}
\int_{\R} & g(v)e^{-iv\lambda} \bra \varphi_0, \cU_v \varphi_0 \ket dv \\
& = \frac{\h T}{4}\left(\widetilde{g}(0)J_0(\frac{2}{\h})+\frac{1}{2}\widetilde{g}''(c)(\frac{\h}{T})^2 J_2(\frac{2}{\h})\right)+O(\h^\infty) \sim C h^{3/2}T
\end{align*}
where $J_0(x)=\int_0^{2\pi} e^{-x(1-\cos\theta)}d\theta$ and 
$$
J_2(x)=\int_0^{2\pi} ((1-\cos\theta)^2-\frac{1}{2}\h\cos\theta)e^{-x(1-\cos\theta)}d\theta.
$$

The result then follows from the asymptotic 
$$
\int_0^{2\pi}\cos(n\theta)e^{-x(1-\cos\theta)}d\theta = 2\pi e^{-x}I_n(x)\sim \sqrt{\frac{2\pi}{x}}~, \quad x\to\infty,
$$
where $I_n(x)$ denotes the modified Bessel function of the first kind.
\end{proof}

If we now apply this asymptotic to the cases $G=\widetilde{H}$, $T^{-1}\widetilde{H}'$, we find
$$
\|(\Delta+\lambda)\Psi_\lambda\|_{L^2(Q)}^2 \sim C_0 \frac{\h^{3/2}}{T}~, \quad
\|\Psi_\lambda\|_{L^2(Q)}^2 \sim C_1 h^{3/2}T.
$$
Recall $T=O(\h^{3/4+\e/2})$. Thus, the maximal time we may choose is of order $T\asymp \h^{3/4+\e/2}$. We conclude
$$
\frac{\|(\Delta+\lambda)\Psi_\lambda\|_{L^2(Q)}}{\|\Psi_\lambda\|_{L^2(Q)}}=O(\frac{1}{T})=O(\h^{-3/4-\e/2})=O(\lambda^{3/8+\e/4}).
$$

Note that the condition $\frac{T}{\h}\leq L$ together with $T\asymp \h^{3/4+\e/2}$ implies that $\h^{-1/4+\e/2}=\lambda^{1/8-\e/4}=O(L)$, i.e. $\lambda=O(L^{8/(1-2\e)})$. So given a metric cylinder of length $L$ we may construct a continuous family of quasimodes with quasienergy $\lambda\in(0,cL^{8/(1-2\epsilon)}]$ and spectral width $O(\lambda^{3/8+\e/4})$.

\section{Semi-classical measures}

We may now construct a family of quasimodes $\{\Psi_\lambda\}_{\lambda>0}$ as follows. We denote by $\Sigma\subset\Sch^1$ the set of rational momentum vectors associated with the set of periodic metric cylinders in $Q$. Moreover, we denote by $\cL$ the length spectrum, i.e. the set of lengths of periodic metric cylinders in $Q$. We have (cf. \cite{BGKT98}) that $\Sigma$ is dense in $\Sch^1$ and
$$
\#\{l\in\cL \mid l\leq T\}\asymp T^2.
$$

In view of the counting asymptotic for the length spectrum above, given any $\xi_0\in\Sch^1$, there exists a sequence $(\xi_k)_{k\in\N}\subset\Sigma$, $\xi_k\rightarrow\xi_0$, so that the sequence of lengths of these cylinders $(L_k)_{k\in\N}$ is an increasing sequence and $L_k\nearrow +\infty$. In such a way we obtain a sequence of intervals $I_k:=]cL_k^{8/(1-2\epsilon)},cL_{k+1}^{8/(1-2\epsilon)}]$, which cover the half-line $\R_+$, such that $\Psi_\lambda$, $\lambda\in I_k$, is a quasimode with spectral width $O(\lambda^{3/8+\epsilon})$. Note that the quasimodes $\Psi_\lambda$ with $\lambda\in I_k$ are associated with a momentum vector $\xi_k$ and as $\lambda\to\infty$ (i.e. $k\to\infty$) we have $\xi_k\to\xi_0$.

We recall the formula for the restriction of the Wigner distribution to momentum space
$$
d\mu_{\Phi_\lambda}(\xi)=\h^{-2}\|\Phi_\lambda\|_2^{-2}\left|\widehat{\Phi_\lambda}(\frac{\xi}{\h})\right|^2 d\xi.
$$

A calculation confirms 
$$
\widehat{\Phi_\lambda}(\frac{\xi}{\h})=\h\int_\R H(t)e^{\frac{it}{\h^2}(1-|\xi|^2)}\widehat{\gamma}(\frac{\xi-\xi_k}{\h^{1/2}})dt
$$
which implies that almost all of the mass of the measure $d\mu_{\Phi_\lambda}(\xi)$ is concentrated in a ball $B(\xi_k,\h^{1/2-\e})$, because of the rapid decay of $\widehat{\gamma}$. It follows that 
$$
\lim_{\h\to 0}\frac{\bra\Op(a)\Phi_\lambda,\Phi_\lambda\ket_{L^2(\R^2)}}{\|\Phi_\lambda\|_2^2}=a(\xi_0) .
$$

Indeed, it follows from the construction of the inital state $\psi_0$, which is simply a cutoff of the Gaussian state $\varphi_0$ on a ball of size $\h^{1/2-\e}$, and almost all of whose mass, up to a proportion of $O(\h^\infty)$, is microlocally supported on the set $B(x_0,\h^{1/2-\e})\times B(\xi_k,\h^{1/2-\e})$, that
$$
\bra\Op(a)\Psi_\lambda,\Psi_\lambda\ket_{L^2(Q)}=\bra\Op(a)\Phi_\lambda,\Phi_\lambda\ket_{L^2(\R^2)}+O(\h^\infty)
$$
and therefore 
$$
\lim_{\h\to 0}\frac{\bra\Op(a)\Psi_\lambda,\Psi_\lambda\ket_{L^2(Q)}}{\|\Psi_\lambda\|_2^2}=a(\xi_0)
$$
which implies, as $\lambda\to\infty$,
$$
d\mu_{\Psi_\lambda}(\xi)\xrightarrow{w*}\delta(\xi-\xi_0) .
$$

This concludes the proof of Theorem \ref{thm-Q}.

\section{Quasimodes on rational polygons}

In order to verify that $\Psi_\lambda^P$ is indeed a quasimode for the Laplacian on $P$, we calculate its $L^2$-norm
\begin{align*}
& \|\sum_{g\in D}\Psi_\lambda(gx)\|_{L^2(P)}^2 = \sum_{g,g'\in D} \int_P\Psi_\lambda(gx)\overline{\Psi_\lambda(g'x)}dx\\
& = \sum_{g\in D}\int_P |\Psi_\lambda(gx)|^2 dx +\sum_{g\neq g'} \int_P\Psi_\lambda(gx)\overline{\Psi_\lambda(g'x)}dx \\
& = \int_Q |\Psi_\lambda(x)|^2 dx +O(\h^\infty).
\end{align*}

Moreover, since $\Delta(\Psi_\lambda\circ g)=(\Delta\Psi_\lambda)\circ g$, we find
$$
\|(\Delta+\lambda)\Psi_\lambda^P\|_{L^2(P)}=\|(\Delta+\lambda)\Psi_\lambda\|_{L^2(Q)}+O(\h^\infty)
$$
which implies that $\Psi_\lambda^P$ is a quasimode of the Laplacian on $P$ with spectral width $O(\lambda^{3/8+\e})$.

Let us now calculate the restriction of the semi-classical measures to momentum space, as $\lambda\to\infty$,
\begin{align*}
& \frac{\bra\Op(a)\Psi_\lambda^P,\Psi_\lambda^P\ket_{L^2(P)}}{\|\Psi_\lambda\|_{L^2(P)}^2}
= \sum_{g,g'\in D}\frac{\bra\Op(a)\Psi_\lambda(gx),\Psi_\lambda(g'x)\ket_{L^2(P)}}{\|\Psi_\lambda\|_{L^2(P)}^2}\\
= & \sum_{g\in D} \frac{\bra\Op(a)\Psi_\lambda(gx),\Psi_\lambda(gx)\ket_{L^2(P)}}{\|\Psi_\lambda\|_{L^2(Q)}^2}+O(\h^\infty)\\
= & \frac{1}{\|\Psi_\lambda\|_{L^2(Q)}^2}\sum_{g\in D}\int_{S^*P}a(\xi)d\mu_{\Psi_\lambda}(gx,g\xi)+O(\h^\infty)\\
= & \frac{\vol(P)}{\vol(Q)}\sum_{g\in D}\int_{\Sch^1}a(\xi)d\mu_{\Psi_\lambda}(g\xi)+O(\h^\infty)\\
= & \frac{1}{|D|}\sum_{g\in D}\int_{\Sch^1}a(\xi)d\mu_{\Psi_\lambda}(g\xi)+O(\h^\infty) \to \frac{1}{|D|}\sum_{g\in D}a(g\xi_0) ,
\end{align*}
which yields Corollary \ref{cor-P}.

\end{document}